\documentclass[11pt,letterpaper]{article}
\usepackage{amsthm,color,latexsym,graphicx,url}
\usepackage[margin=1in]{geometry}
\usepackage[font=small,labelfont=bf]{caption}
\usepackage[labelformat=simple]{subcaption}

\usepackage{libertine}\usepackage[libertine]{newtxmath}
\usepackage[scaled=0.96]{zi4}
\usepackage[utf8]{inputenc}
\usepackage{graphicx}
\usepackage[dvipsnames]{xcolor}
\usepackage{tikz}
\usepackage{pgffor}

\colorlet{wheel}{Purple!50}
\colorlet{shuriken}{ForestGreen!50}
\colorlet{staple}{orange!50}
\colorlet{interior}{red!20}
\colorlet{exterior}{blue!20}

\colorlet{clean}{RedViolet!80}
\def\CLEAN#1{\textcolor{clean}{#1}}

\title{Tiling with Three Polygons is Undecidable
  }
\author{%
  Erik D. Demaine%
    \thanks{Computer Science and Artificial Intelligence Laboratory, Massachusetts Institute of Technology, USA.
      \protect\url{edemaine@mit.edu}}
\and
  Stefan Langerman%
    \thanks{Computer Science Department, Universit\'e libre de Bruxelles, Belgium.
      \protect\url{stefan.langerman@ulb.ac.be}}
}
\date{}

\newif\ifabstract
\abstracttrue
\newif\iffull
\ifabstract \fullfalse \else \fulltrue \fi

\usepackage{hyperref}
\hypersetup{breaklinks,bookmarks,bookmarksnumbered,bookmarksopen,bookmarksopenlevel=2}
{\makeatletter \hypersetup{pdftitle={\@title}}}

{\makeatletter
 \gdef\xxxmark{%
   \expandafter\ifx\csname @mpargs\endcsname\relax 
     \expandafter\ifx\csname @captype\endcsname\relax 
       \marginpar{xxx}
     \else
       xxx 
     \fi
   \else
     xxx 
   \fi}
 \gdef\xxx{\@ifnextchar[\xxx@lab\xxx@nolab}
 \long\gdef\xxx@lab[#1]#2{\textbf{[\xxxmark #2 ---{\sc #1}]}}
 \long\gdef\xxx@nolab#1{\textbf{[\xxxmark #1]}}
}

{\makeatletter \gdef\fps@figure{!htbp}}

\let\realbfseries=\bfseries
\def\bfseries{\realbfseries\boldmath}

\newtheorem{theorem}{Theorem}[section]
\newtheorem{lemma}[theorem]{Lemma}
\newtheorem{corollary}[theorem]{Corollary}
\theoremstyle{definition}
\newtheorem{definition}{Definition}
\newtheorem{property}{Property}

\newtheorem{problem}{Problem}

\let\epsilon=\varepsilon
\def\defn#1{\textbf{\textit{\boldmath #1}}}

\begin{document}
\maketitle

\begin{abstract}
  We prove that the following problem is co-RE-complete and thus undecidable:
  given three simple polygons,
  is there a tiling of the plane where every tile is
  an isometry of one of the three polygons
  (either allowing or forbidding reflections)?
  This result improves on the best previous construction
  which requires five polygons.
\end{abstract}

\begin{quote}\raggedleft
  ``Three Rings for the Elven-kings under the sky, \dots''

  --- J. R. R. Tolkien, \emph{The Lord of the Rings}, epigraph
\end{quote}

\section{Introduction}


A \defn{tiling} of the plane \cite{Gruenbaum-Shephard}
is a covering of the plane by nonoverlapping polygons called \defn{tiles}, 
isometric copies of one or more geometric shapes called \defn{prototiles},
without gaps or overlaps. 
In this paper,
we study the most fundamental computational problem about tilings: 
\begin{problem}[Tiling]
  Given one or more prototiles, can they tile the plane?
\end{problem}


The tiling problem is \defn{undecidable} --- solved by no finite algorithm.
Golomb \cite{Golomb1970} was first to prove this result, by
building $n$ polyominoes that simulate $n$ \defn{Wang tiles}
\cite{Wang1961} --- unit squares with edge colors that must match ---
by adding color-specific bumps and dents to each edge.
Four years earlier, Berger \cite{Berger1966} proved that
tiling with Wang tiles is undecidable
(disproving Wang's original conjecture \cite{Wang1961})
by showing how they can simulate a Turing machine.
Robinson \cite{Robinson1971} later simplified Berger's proof.
The worst-case number $n$ of tiles (Wang or polyomino)
is $\Theta(|Q| \cdot |\Sigma|)$, where $|Q|$ and $|\Sigma|$ are the
number of states and symbols in the simulated Turing machine, respectively.

\paragraph{Constant Number of Prototiles.}
The first constant and previously best upper bound on the number of prototiles
required to make the tiling problem undecidable is $5$,
as proved by Ollinger fifteen years ago \cite{Ollinger09}.
Our main result, proved in Section~\ref{sec:three},
is an improvement of this upper bound to~$3$:

\begin{theorem} \label{thm:intro:main}
  Given three simple-polygon prototiles,
  determining whether they tile the plane is undecidable.
\end{theorem}

It remains open whether tiling with one or two given prototiles is decidable.
\defn{Periodic} tilings (tilings with two translational symmetries)
can be found algorithmically by enumerating fundamental domains,
as we show in Section~\ref{sec:periodic}.
(Surprisingly, this intuitive fact does not seem to have been
explicitly proved before, except in special settings like Wang tiles
\cite{Wang1961}.)
Thus a necessary condition for undecidability is the existence of prototile(s)
with only aperiodic tilings.
Recently, Smith, Myers, Kaplan, and Goodman-Strauss \cite{monotile}
found a single prototile with this property, so
there are no obvious obstacles to undecidability.

\paragraph{Tiling by Translation.}
Our construction relies on rotation of the prototiles
(but works independent of whether we allow reflections).
If we restrict to tiling by translation only,
then Ollinger's construction can be modified to use $11$ prototiles,
by adding some rotations of the five polyominoes \cite{Ollinger09}.
This upper bound was improved to $10$ by Yang \cite{Yang23}
and to $8$ by Yang and Zhang \cite{Yang-Zhang-2024}.
All of these constructions use polyominoes.
In higher dimensions, Yang and Zhang
\cite{yang2024undecidabilitytranslationaltiling4dimensional}
improved the upper bound to five polycube prototiles in 3D,
and four polyhypercube prototiles in 4D.

The tiling-by-translation problem also has
a lower bound of $2$ for undecidability:
any single polygon that tiles the plane by translation
can do so by periodic (even isohedral) tiling \cite{G-BN1991}.
This result also holds for disconnected polyominoes \cite{B20-R2-periodicity}.
If we generalize to tiling a specified periodic subset of $d$-dimensional space,
where $d$ is part of the input, then
Greenfeld and Tao \cite{greenfeld2024undecidabilitytranslationalmonotilings}
recently proved tiling to be undecidable
with a single disconnected polyhypercube.

\paragraph{Periodic Target.}
We show that Greenfeld and Tao's generalization to tiling a specified periodic
subset \cite{greenfeld2024undecidabilitytranslationalmonotilings}
changes the best known results also for undecidability of tiling the plane.
Our $3$-polygon construction and
Ollinger's $5$-polyomino construction \cite{Ollinger09},
and Yang and Zhang's $8$-polyomino translation-only polyomino construction
\cite{Yang-Zhang-2024}
all have one prototile (our shurikens, and their jaws)
that appear periodically in any tiling of the plane.
Thus, if we remove that pattern from the target, we obtain a periodic subset of
the plane which can be tiled using a reduced number of prototiles of
$2$, $4$, and $7$, respectively.
In particular, we prove

\begin{corollary} \label{cor:intro:two}
  Given two simple-polygon prototiles, and
  given a periodic subset of the plane,
  determining whether the two prototiles tile the periodic subset
  is undecidable.
\end{corollary}

\paragraph{Logical Undecidability.}
Algorithmic undecidability implies \defn{logical undecidability}
(as explained in \cite{Greenfeld-Tao-2023} in the context of tilings).
In particular, our result implies that there are three polygon prototiles
that cannot be proved or disproved to tile the plane,
for any fixed set of axioms (e.g., ZFC).
Otherwise, we would obtain a finite algorithm to decide tileability,
by enumerating all proofs.

\begin{corollary}
  For any fixed set of axioms,
  there are three fixed simple-polygon prototiles such that
  both ``these prototiles tile the plane'' and
  ``these prototiles do not tile the plane''
  have no proof.
\end{corollary}

\paragraph{Tiling Completion.}
Undecidability of tiling requires the set of prototiles to depend on the
Turing machine simulation.
To obtain undecidability with a fixed set of prototiles,
we can generalize the tiling problem as follows \cite{Robinson1971}:

\begin{problem}[Tiling Completion]
  Given one or more prototiles, and given some already placed tiles,
  can this placement be extended to a tiling of the plane?
\end{problem}

Robinson \cite{Robinson1971} gave the first result on this problem:
a set of 36 prototiles (Wang tiles or polygons)
for which tiling completion is undecidable.
This result applies the general Turing machine simulation to
Minsky's 4-symbol 7-state universal Turing machine,
so only a finite number of tiles need to be preplaced
to represent the Turing machine to simulate.
Likely this result could be improved using newer smaller
universal Turing machines \cite{Woods-Neary-2009}.
If we allow for (countably) infinitely many tiles to be preplaced,
we can use semi-universal Turing machines and simulate Rule 110,
enabling undecidability with just six supertiles (Wang tile or polygons)
\cite{Yang-2013-thesis}.

Our main result reduces this upper bound to $3$,
in the stronger model of finitely many preplaced tiles:

\begin{corollary} \label{cor:intro:completion}
  There are three fixed simple-polygon prototiles such that,
  given a finite set of already placed tiles,
  determining whether this placement can be extended to a tiling of the plane
  is undecidable.
\end{corollary}

\paragraph{Co-RE-completeness.}
While past results on tiling and tiling completion have focused on
undecidability, all such proofs actually show \defn{co-RE-hardness}:
the simulated Turing machine halts if and only if the prototiles fail to tile.
Co-RE-hardness is a more precise statement than undecidability,
so we use that phrasing here.
But it also raises the question: are tiling and tiling completion in co-RE?
Surprisingly, this question does not seem to have been solved (or even asked)
in the literature before.
In Section~\ref{sec:co-RE}, we prove that the answer is ``yes'':

\begin{theorem} \label{thm:intro:coRE}
  Given a finite set of polygon prototiles,
  and given a (possibly empty) connected set of already placed tiles,
  determining whether this placement can be extended to a tiling the plane
  is in co-RE.
\end{theorem}

This result holds in a very general model for polygons:
the angles and edge lengths can be represented as \defn{computable} numbers,
(meaning that a Turing machine can output the first $n$ bits, given~$n$).
Our three-polygon construction uses a more restricted model,
where the angles are rational multiples of $\pi$
and the edge lengths are constant-size radical expressions, showing the problem to be co-RE-complete for every model in between.

\begin{corollary}[Stronger form of Theorem~\ref{thm:intro:main}] \label{cor:intro:complete}
  Given three simple-polygon prototiles, where the angles and edge lengths are
  specified by computable numbers or by constant-size radical expressions,
  determining whether they tile the plane is co-RE-complete.
\end{corollary}

\section{Wang Tiling: Signed and Unsigned}

We reduce from Wang tiling, which is known to be undecidable.
A \defn{Wang tile} is a square with a \defn{glue} on each edge.
Classically, Wang tiles are \defn{unsigned},
meaning that glues match if they are equal, and
\defn{translation-only}, meaning they have a specified orientation
of which edge is north, east, south, and west.
In 1966, Berger proved unsigned Wang tiling undecidable:

\begin{theorem}[{\cite{Berger1966}}]
  Given a set of Wang tiles, it is undecidable to determine
  whether they tile the plane by translation only,
  matching glues of equal value.
\end{theorem}

Our first reduction converts unsigned Wang tiling to a variant
that is \defn{signed}, meaning every glue has a sign ($+$ or $-$)
and value, and glues match if they have opposite sign and equal value,
and \defn{free}, meaning the tile can be rotated and/or reflected.
This result was proved by Robinson \cite[p.~179]{Robinson1971};
we give the simple reduction here for completeness.

\begin{lemma}
  Given a set of unsigned translation-only Wang tiles $T$,
  we can construct a set of signed free Wang tiles $T'$
  that has the same tilings as $T$ up to global isometry,
  allowing or forbidding reflection.
\end{lemma}

\begin{proof}
  First, modify $T$ so that every west/east glue is different
  from every south/north glue.
  This transformation at most doubles the number of glues,
  and does not change the set of translation-only Wang tilings
  because these glues cannot interact via translations.

  Second, make the tiles signed by
  adding a sign of $-$ on every west and south edge glue,
  and adding a sign of $+$ on every north and east edge glue.
  Again this transformation at most doubles the number of glues,
  and does not change the set of translation-only Wang tilings
  because we always combine west glues with east glues
  and south glues with north glues, which have opposite signs by definition.

  We claim that the resulting tile set $T'$ prevents rotations and reflections
  relative to a single root tile.
  Define compass directions (west/east/south/north) according to this root tile.
  The west edge of the root tile is negative and among the west/east glues,
  so the only way a tile can attach on this side is via a positive glue among
  the west/east glues, which implies it is the east edge of another tile,
  forcing the west neighbor tile to have the same orientation.
  The same statement holds for the other compass directions,
  forcing all tile orientations to match the root tile.
\end{proof}

Henceforth when we say ``Wang tiles'' we mean signed free Wang tiles.

\section{Three Tiles That Simulate $n$ Signed Wang Tiles}
\label{sec:three}

We implement any set of $n$ Wang tiles with three tiles,
illustrated in Figure~\ref{fig:three tiles}:
\begin{enumerate}
\item the \emph{wheel} which encodes all of the Wang tiles,
\item the \emph{staple} which covers the unused Wang tiles of each wheel, and
\item the \emph{shuriken} which fills in the remaining gaps.
\end{enumerate}


\usetikzlibrary{math}
\tikzmath{
  real \l, \eps1, \eps2, \alph, \beth;
  integer \k;
  \l = 1.5;
  \k = 3;
  \eps2 = 5;
  \eps3 = 5;
  \eps1 = \eps2+\eps3;
}

\usetikzlibrary{turtle}
\usetikzlibrary{angles,quotes,backgrounds}

\def\tweedledee{[turtle={
  forward=\l*(2-cos(\eps3)+sin(\eps1)),
  left=(90+\eps1), forward=\l,
  right=90+\eps2, fd=\l,
  rt=90+\eps3,forward=2*\l*(cos(\eps1)+sin(\eps3)),
  lt=90+\eps3,fd=\l,
  lt=90+\eps2,fd=\l,
  rt=90+\eps1,fd=\l*(2-cos(\eps3)+sin(\eps1))}]}

\def\tweedledum{[turtle={
  forward=\l*(2-cos(\eps3)+sin(\eps1)),
  rt=(90+\eps1), forward=\l,
  lt=90+\eps2, fd=\l,
  lt=90+\eps3,forward=2*\l*(cos(\eps1)+sin(\eps3)),
  rt=90+\eps3,fd=\l,
  rt=90+\eps2,fd=\l,
  lt=90+\eps1,fd=\l*(2-cos(\eps3)+sin(\eps1))}]}


\def\notch{[turtle={
  forward=\l*(2-cos(\eps3)+sin(\eps1)),
  left=(90+\eps1), forward=\l,
  right=90+\eps2, fd=\l,
  rt=2*\eps3, fd=\l,
  rt=90+\eps2,fd=\l,
  lt=90+\eps1,fd=\l*(2-cos(\eps3)+sin(\eps1)) 
}]}

\def\staple{[turtle={
  lt=90+\eps1, forward=\l,
  right=90+\eps2, fd=\l,
  rt=90+\eps3,forward=2*\l*(cos(\eps1)+sin(\eps3)),
  rt=90+\eps3,fd=\l,
  rt=90+\eps2,fd=\l
}]}

\def\wheel{
  \foreach \j in {0,1} {
    \foreach \i in {1,...,\k} {
      [turtle={forward=5*\l}]
      \tweedledee
      [turtle={forward=4*\l}]
      \tweedledum
      [turtle={forward=5*\l, right=90/\k}]      
    }
    \foreach \i in {1,...,\k} {
      [turtle={forward=5*\l}]
      \tweedledum
      [turtle={forward=4*\l}]
      \tweedledee
      [turtle={forward=5*\l, right=90/\k}]      
    }
  }
}

\def\shuriken{
  \foreach \j in {1,...,4} {
    \foreach \i in {2,...,\k} {
      [turtle={rt=90/\k,forward=5*\l}]
      \notch
      [turtle={forward=4*\l}]
      \notch
      [turtle={forward=5*\l}]      
    }
    [turtle={lt=180-90/\k}]
  }
}

\begin{figure}
  \tikzmath{
    \l = 1/10;
  }
  \centering
  \subcaptionbox{Wheel\label{fig:wheel}}{%
    \centering
    \begin{tikzpicture}[turtle/distance=2cm]
      \filldraw [fill=wheel, turtle={home, rt}]
      \wheel ;
    \end{tikzpicture}
  }%
  \hfill
  \subcaptionbox{Shuriken\label{fig:shuriken}}{%
    \centering
    \begin{tikzpicture}[turtle/distance=2cm]
      \filldraw [fill=shuriken, turtle={home, rt}]
      \shuriken ;
    \end{tikzpicture}
  }%
  \hfill
  \subcaptionbox{\raggedright Staple\label{fig:staple small}}[3.5em]{%
    \begin{tikzpicture}[turtle/distance=2cm]
      \filldraw [fill=staple] [turtle={home, lt=\eps1}] coordinate (A) coordinate (B)
      [turtle={forward=\l}] coordinate (C)
      [turtle={right=90+\eps2, fd=\l}] coordinate (D)
      [turtle={rt=90+\eps3,forward=2*\l*(cos(\eps1)+sin(\eps3))}] coordinate (E)
      [turtle={rt=90+\eps3,fd=\l}] coordinate (F)
      [turtle={rt=90+\eps2,fd=\l}] ;
    \end{tikzpicture}
  }%
  \caption{The three tiles in our construction, to scale;
    Figure~\ref{fig:three tiles zoomed} shows zoomed details
    of the construction.
    The wheel is just an example; it depends on the $n$ Wang tiles being simulated.
    The shuriken depends (only) on~$n$.}
  \label{fig:three tiles}
\end{figure}
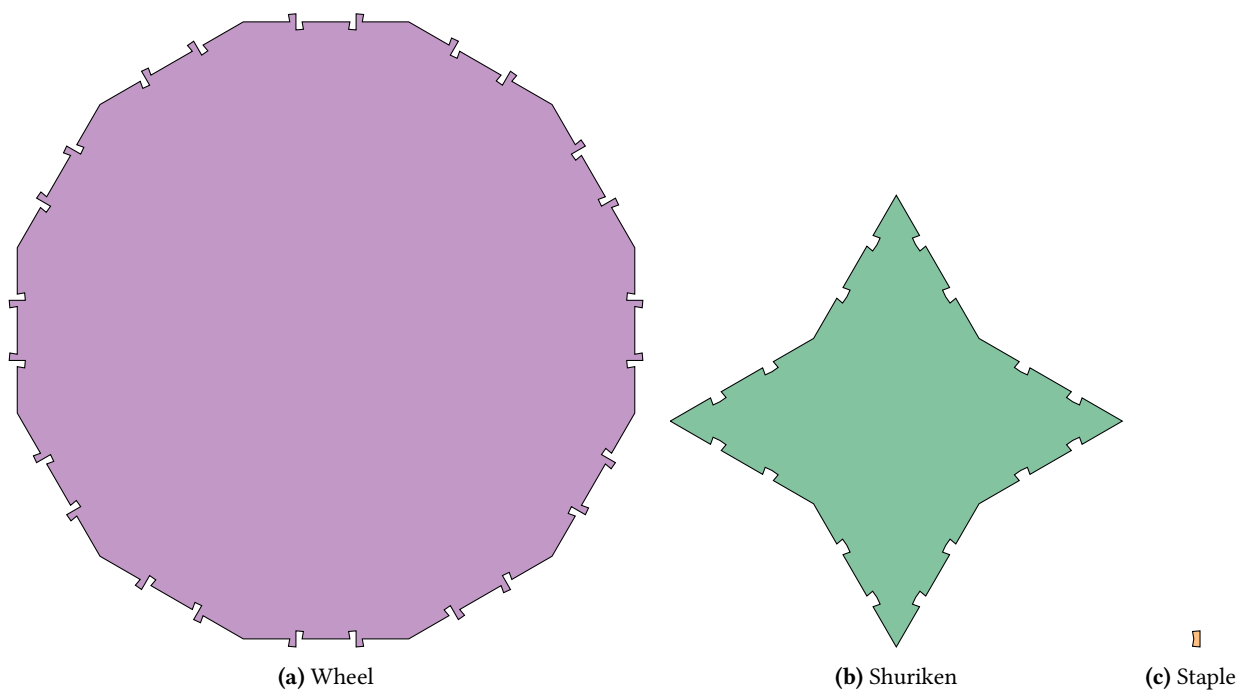

\begin{figure}[h!]
  \tikzmath{
    \l = 1.5;
  }
  \centering
  \subcaptionbox{Wheel: Tweedledee ($0$)\label{fig:tweedledee}}{%
    \centering
    \begin{tikzpicture}[turtle/distance=2cm, fill=wheel]
      \draw [turtle={home, rt}] coordinate (A)
      [turtle={forward=\l*(2-cos(\eps3)+sin(\eps1))}] coordinate (B)
      [turtle={left=(90+\eps1), forward=\l}] coordinate (C)
      pic ["$\alpha$", draw, fill=interior] {angle=C--B--A}
      [turtle={right=90+\eps2, fd=\l}] coordinate (D)
      pic ["$\beta$", draw, fill=exterior] {angle=B--C--D}
      [turtle={rt=90+\eps3,forward=2*\l*(cos(\eps1)+sin(\eps3))}] coordinate (E)
      pic ["$\beta$", draw, fill=exterior] {angle=C--D--E}
      [turtle={lt=90+\eps3,fd=\l}] coordinate (F)
      pic ["$\beta$", draw, fill=interior] {angle=F--E--D}
      [turtle={lt=90+\eps2,fd=\l}] coordinate (G)
      pic ["$\beta$", draw, fill=interior] {angle=G--F--E}
      [turtle={rt=90+\eps1,fd=\l*(2-cos(\eps3)+sin(\eps1))}] coordinate (H)
      pic ["$\alpha$", draw, fill=exterior] {angle=F--G--H};
      \begin{pgfonlayer}{background}
        \fill [wheel] (A) -- (B) -- (C) -- (D) -- (E) -- (F) -- (G) -- (H) -- ++(0,-1.5*\l) -- ++(-4*\l,0cm) -- cycle;
      \end{pgfonlayer}
    \end{tikzpicture}
  }%
  \hfil
  \subcaptionbox{Wheel: Tweedledum ($1$)\label{fig:tweedledum}}{%
    \centering
    \begin{tikzpicture}[turtle/distance=2cm, fill=wheel]
      \draw [turtle={home, rt}] coordinate (A)
      [turtle={forward=\l*(2-cos(\eps3)+sin(\eps1))}] coordinate (B)
      [turtle={rt=(90+\eps1), forward=\l}] coordinate (C)
      pic ["$\alpha$", draw, fill=exterior] {angle=A--B--C}
      [turtle={lt=90+\eps2, fd=\l}] coordinate (D)
      pic ["$\beta$", draw, fill=interior] {angle=D--C--B}
      [turtle={lt=90+\eps3,forward=2*\l*(cos(\eps1)+sin(\eps3))}] coordinate (E)
      pic ["$\beta$", draw, fill=interior] {angle=E--D--C}
      [turtle={rt=90+\eps3,fd=\l}] coordinate (F)
      pic ["$\beta$", draw, fill=exterior] {angle=D--E--F}
      [turtle={rt=90+\eps2,fd=\l}] coordinate (G)
      pic ["$\beta$", draw, fill=exterior] {angle=E--F--G}
      [turtle={lt=90+\eps1,fd=\l*(2-cos(\eps3)+sin(\eps1))}] coordinate (H)
      pic ["$\alpha$", draw, fill=interior] {angle=H--G--F};
      \begin{pgfonlayer}{background}
        \fill [wheel] (A) -- (B) -- (C) -- (D) -- (E) -- (F) -- (G) -- (H) -- ++(0,-1.5*\l) -- ++(-4*\l,0cm) -- cycle;
        \end{pgfonlayer}
    \end{tikzpicture}
  }%

  \medskip

  \subcaptionbox{Shuriken: Notch\label{fig:notch}}{%
    \centering
    \begin{tikzpicture}[turtle/distance=2cm]
      \draw [turtle={home, rt}] coordinate (A)
      [turtle={forward=\l*(2-cos(\eps3)+sin(\eps1))}] coordinate (B)
      [turtle={left=(90+\eps1), forward=\l}] coordinate (C)
      pic ["$\alpha$", draw, fill=exterior] {angle=C--B--A}
      [turtle={right=90+\eps2, fd=\l}] coordinate (D)
      pic ["$\beta$", draw, fill=interior] {angle=B--C--D}
      [turtle={rt=2*\eps3, fd=\l}] coordinate (E)
      pic ["$2\beta$", draw, fill=interior] {angle=C--D--E}
      [turtle={rt=90+\eps2,fd=\l}] coordinate (F)
      pic ["$\beta$", draw, fill=interior] {angle=D--E--F}
      [turtle={lt=90+\eps1,fd=\l*(2-cos(\eps3)+sin(\eps1))}] coordinate (G)
      pic ["$\alpha$", draw, fill=exterior] {angle=G--F--E};
      \begin{pgfonlayer}{background}
        \fill [shuriken] (A) -- (B) -- (C) -- (D) -- (E) -- (F) -- (G) -- ++(0,1.5*\l) -- ++(-4*\l,0cm) -- cycle;
      \end{pgfonlayer}
    \end{tikzpicture}
  }%
  \hfill
  \subcaptionbox{Staple\label{fig:staple}}{%
    \centering

    \begin{tikzpicture}[turtle/distance=2cm]
      \draw [thick] [turtle={home, lt=\eps1}] coordinate (B)
      [turtle={forward=\l}] coordinate (C)
      [turtle={right=90+\eps2, fd=\l}] coordinate (D)
      pic ["$\beta$", draw, fill=exterior] {angle=B--C--D}
      [turtle={rt=90+\eps3,forward=2*\l*(cos(\eps1)+sin(\eps3))}] coordinate (E)
      pic ["$\beta$", draw, fill=exterior] {angle=C--D--E}
      [turtle={rt=90+\eps3,fd=\l}] coordinate (F)
      pic ["$\beta$", draw, fill=exterior] {angle=D--E--F}
      [turtle={rt=90+\eps2,fd=\l}] 
      pic ["$\beta$", draw, fill=exterior] {angle=E--F--A}
      pic ["$2\alpha$", draw, fill=interior] {angle=C--A--F};
      \begin{pgfonlayer}{background}
        \fill [staple] (B) -- (C) -- (D) -- (E) -- (F) -- cycle;
      \end{pgfonlayer}
    \end{tikzpicture}
  }%
  \hfill
  \subcaptionbox{Combining the tweedle, notch, and staple.\label{fig:threefigs}}{%
    \centering
    \begin{tikzpicture}[turtle/distance=2cm]
      \draw [thick] [turtle={home,rt}] \tweedledum;
      \fill [fill=wheel] [turtle={home,rt}] \tweedledum  -- ++(0,-1.5*\l) -- ++(-4*\l,0cm) -- cycle;
      \fill [fill=shuriken] [turtle={home,rt}] \notch  -- ++(0,1.5*\l) -- ++(-4*\l,0cm) -- cycle; 
      \draw [thick] [turtle={home,rt}] \notch; 
      \filldraw [thick,fill=staple] [turtle={home,rt,forward=\l*(2-cos(\eps3)+sin(\eps1))}] \staple;
    \end{tikzpicture}
  }%
  \caption{Zoomed views of portions of the three tiles in our construction
    ($15\times$ scale compared to Figure~\ref{fig:three tiles}).}
  \label{fig:three tiles zoomed}
\end{figure}
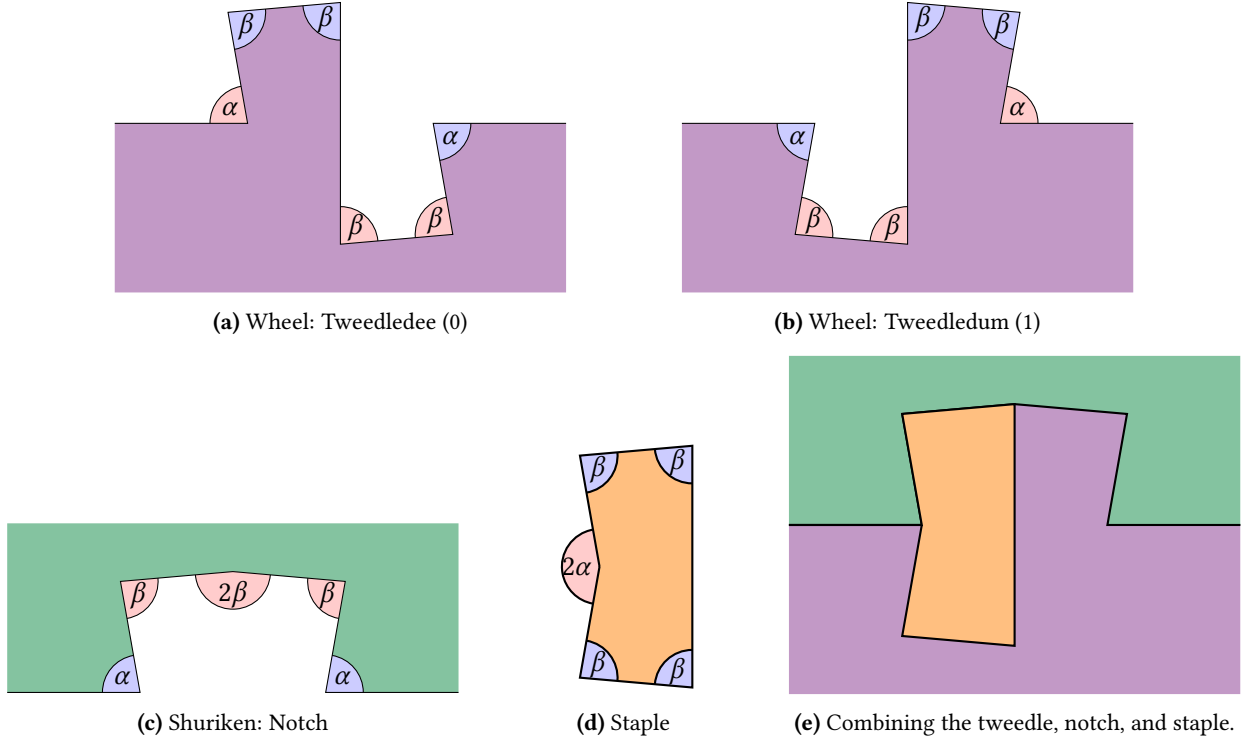

\def\tweedleword{
  \tweedledee
  \tweedledee
  \tweedledum
  \tweedledee
  \tweedledum
  \tweedledee
  \tweedledum
}

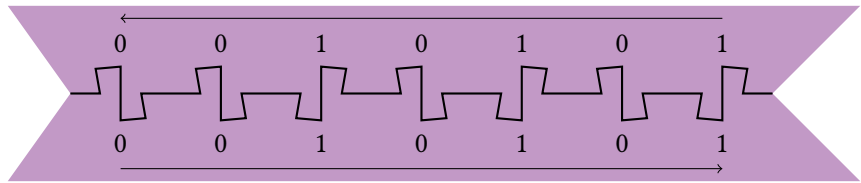
\begin{figure}[ht]
  \centering
  \tikzmath{
    \l = 1/3;
  }
  \begin{tikzpicture}[turtle/distance=2cm]
    \draw [thick] [turtle={home,rt}] \tweedleword;
    \fill [fill=wheel] [turtle={home,rt}] \tweedleword  -- ++(3.5*\l,-3.5*\l) -- ++(-34*\l,0cm) -- cycle;
    \fill [fill=wheel] [turtle={home,rt}] \tweedleword  -- ++(3.5*\l,3.5*\l) -- ++(-34*\l,0cm) -- cycle;
    \draw [thick] [turtle={home,rt}] \tweedleword;
    \draw [->] (2*\l,-3*\l) -- ++(4*6*\l,0);
    \path (2*\l,-2*\l) node {0}
      -- ++(4*\l,0) node {0}
      -- ++(4*\l,0) node {1}
      -- ++(4*\l,0) node {0}
      -- ++(4*\l,0) node {1}
      -- ++(4*\l,0) node {0}
      -- ++(4*\l,0) node {1};
    \draw [->] (26*\l,3*\l) -- ++(-4*6*\l,0);
    \path (2*\l,2*\l) node {0}
      -- ++(4*\l,0) node {0}
      -- ++(4*\l,0) node {1}
      -- ++(4*\l,0) node {0}
      -- ++(4*\l,0) node {1}
      -- ++(4*\l,0) node {0}
      -- ++(4*\l,0) node {1};
  \end{tikzpicture}
  \caption{Matching a glue (top) and its negative (bottom) between two wheels.}
  \label{fig:twiddleword}
\end{figure}

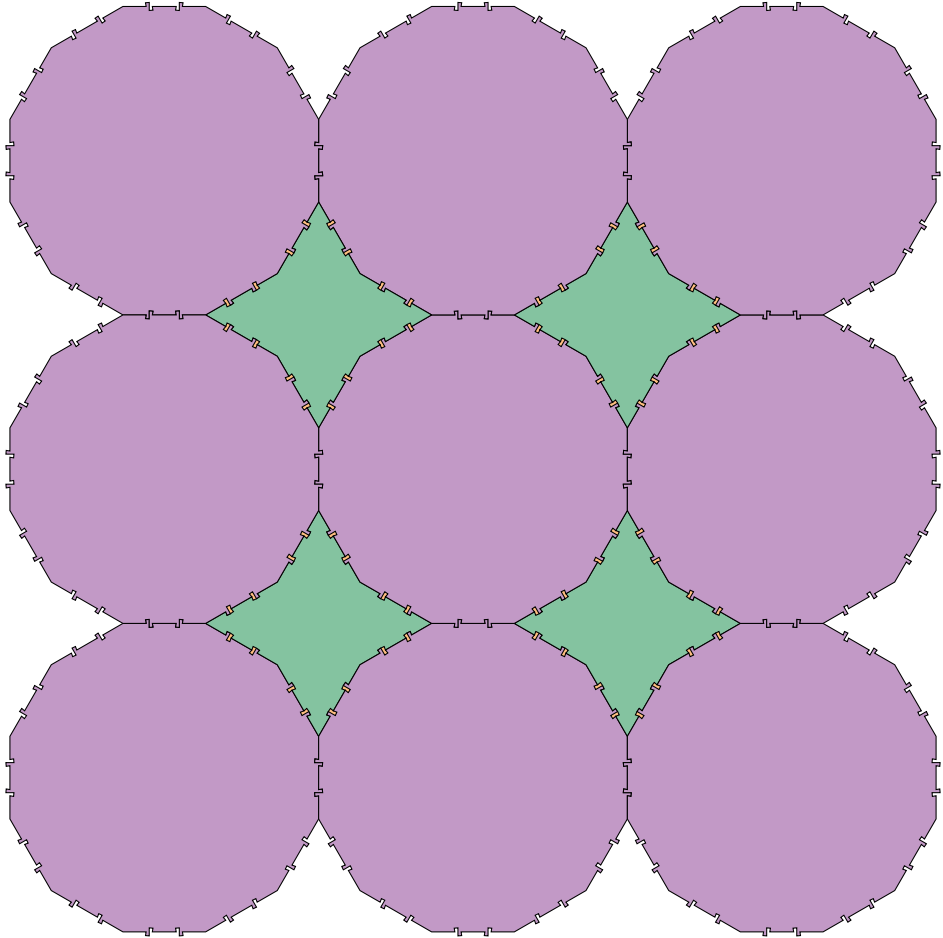
\begin{figure}[ht]
  \centering
  \tikzmath{
    \l = 1/20;
  }
\begin{tikzpicture}[turtle/distance=2cm]
  \fill [turtle=home] (0,-52*\l) [fill=staple]
  \foreach \i in {1,...,4} {
    [turtle={forward=185*\l, right=90}]
  };
  \filldraw [fill=wheel, turtle={home, rt}]
  \wheel;
  \filldraw (22*\l,0) [fill=shuriken, turtle={rt=0}]
  \shuriken ;
  \path [turtle={home,rt} ]\foreach \i in {1,...,\k} {
    [turtle={forward=22*\l, right=90/\k}]
  }
  [turtle={forward=22*\l}] coordinate (A)
  [turtle={rt=180} ]\foreach \i in {1,...,\k} {
    [turtle={forward=22*\l, right=90/\k}]
  } [turtle={forward=22*\l}] coordinate (B)
  [turtle={fd=-22*\l}]
  \foreach \i in {1,...,\k} {
    [turtle={forward=22*\l, right=90/\k}]
  } [turtle={forward=22*\l}]
  [turtle={rt=180} ]\foreach \i in {1,...,\k} {
    [turtle={forward=22*\l, right=90/\k}]
  } coordinate (C)
  [turtle={fd=22*\l}]
  [turtle={rt=180} ]\foreach \i in {1,...,\k} {
    [turtle={forward=22*\l, right=90/\k}]
  } coordinate (D)
  [turtle={fd=22*\l}] coordinate (E)
  [turtle={fd=-22*\l}]
  \foreach \i in {1,...,\k} {
    [turtle={forward=22*\l, right=90/\k}]
  } [turtle={fd=22*\l}] coordinate (F)
  [turtle={rt=180}]\foreach \i in {1,...,\k} {
    [turtle={forward=22*\l, right=90/\k}]
  } [turtle={fd=22*\l}] coordinate (G)
  [turtle={rt=180}]\foreach \i in {1,...,\k} {
    [turtle={forward=22*\l, right=90/\k}]
  } [turtle={fd=22*\l}] coordinate (H)
  [turtle={fd=-22*\l}]
  \foreach \i in {1,...,\k} {
    [turtle={forward=22*\l, right=90/\k}]
  } [turtle={fd=22*\l}] 
  [turtle={rt=180}]\foreach \i in {1,...,\k} {
    [turtle={forward=22*\l, right=90/\k}]
  } coordinate (I)
  [turtle={forward=22*\l, rt=180}]\foreach \i in {1,...,\k} {
    [turtle={forward=22*\l, right=90/\k}]
  } 
  coordinate (J)
;
  \filldraw [turtle=home] (A) [fill=wheel]
  \wheel; 
  \filldraw [turtle=home] (B) [fill=shuriken, turtle={rt}]
  \shuriken ;
  \filldraw [turtle=home] (C) [fill=wheel, turtle={rt}]
  \wheel; 
  \filldraw [turtle=home] (B) [fill=wheel, turtle={lt}]
  \wheel; 
  \filldraw [turtle=home] (D) [fill=wheel]
  \wheel;
  \filldraw [turtle=home] (E) [fill=shuriken]
  \shuriken ;
  \filldraw [turtle=home] (F) [fill=wheel, turtle={lt}]
  \wheel;
  \filldraw [turtle=home] (G) [fill=wheel, turtle={lt=180}]
  \wheel;
  \filldraw [turtle=home] (H) [fill=shuriken, turtle={lt}]
  \shuriken ;
  \filldraw [turtle=home] (I) [fill=wheel, turtle={lt=90}]
  \wheel; 
  \filldraw [turtle=home] (J) [fill=wheel, turtle={lt=180}]
  \wheel; 
\end{tikzpicture}
\caption{Example tiling with the wheel, shuriken, and staple.}
\label{fig:full}
\end{figure}

\tikzmath{
  \l = 2;
}

\subsection{Construction and Intended Tiling}

Suppose we are given a set of $n$ Wang tiles, where the $i$th tile
($1 \leq i \leq n$) has signed glues $n_i,e_i,s_i,w_i$ on its north, east, south, and west edges respectively.
Assume $n$ is an odd integer $\geq 5$ by possibly adding duplicate tiles.

The \defn{wheel} is a regular $4n$-gon with each edge adorned by bumps and notches representing the $4n$ glues.
For tile $i$, the glues $n_i,e_i,s_i,w_i$ adorn sides $i, n+i, 2n+i, 3n+i$ of the $4n$-gon, respectively.
To encode a glue, we encode its value in binary using $b=O(\log n)$ bits, prepend a $00$ at the beginning, and append $01$ at the end.
For negative glues, we reverse the order of the bits, which puts a $10$ at the beginning and a $00$ at the end.
Then we represent each bit with a \defn{tweedledee} ($0$) or \defn{tweedledum} ($1$) gadget, which are rotationally symmetric zig-zags shown in Figures~\ref{fig:tweedledee} and \ref{fig:tweedledum}.
Both follow the sequence of angles $\alpha, \beta, \beta, \beta, \beta, \alpha$
(where $\alpha$ and $\beta$ are defined below);
for tweedledee, this sequence measures
defect, angle, angle, defect, defect, angle, respectively;
while for tweedledum, this measures the opposite
(angle, defect, defect, angle, angle, defect).%
\footnote{%
  It is also possible to use two different convex angles $\beta_1,\beta_2$
  in place of each repetition $\beta,\beta$, but the notation is messier,
  so we opt for this simpler construction.
}
As shown in Figure~\ref{fig:twiddleword}, two adjacent glues match exactly if and only if they have the same value and opposite sign
(where the opposite sign is enforced by the $00$ and $01$ at either end).
This representation also ensures that reflecting a wheel will produce reflected glues that do not match unreflected glues: a reflection causes the bits of a glue to be reversed and negated, so the reflection of a positive glue starts with $01$ and ends with $11$, and the reflection of a negative glue starts with $11$ and ends with $10$, both of which are incompatible with unreflected glues.

By this construction, rotating the wheel so that its $i$th side is horizontal and at the top will have its north, east, south, and west sides represent the glues $n_i,e_i,s_i,w_i$ of tile $i$.
Given a tiling of the plane using this set of Wang tiles, we can place copies of the rotated wheel exactly as in the Wang tiling, and the glues will match exactly. Some space remains between the wheels, which we fill with ``staples'' and ``shurikens''.
See Figure~\ref{fig:full}.

The \defn{shuriken} is composed of four regular concave chains of $n-1$ sides,
matching the lengths and complementary to the angles of the regular $4n$-gon.
Each side is adorned with $b$ reflectionally symmetric \defn{notches},
shown in Figure~\ref{fig:notch},
each consisting of convex angle $\alpha$;
reflex deficits $\beta, 2 \beta, \beta$; and convex angle~$\alpha$.
As shown in Figure~\ref{fig:threefigs}, each notch can fit a tweedle of
either kind, leaving a space that is filled exactly by a \defn{staple}
(shown in Figure~\ref{fig:staple}, and consisting of convex angles
$\beta,\beta,\beta,\beta$ and reflex deficit $2 \alpha$).
Thus each side of the shuriken can exactly match any glue, effectively hiding the unused tiles of each wheel (the glues that are not on the north, east, south, or west sides).


Thus we have shown:
\begin{lemma} \label{lem:build Wang}
  Given a set of $n$ Wang tiles and a tiling of the plane with them,
  we can construct a tiling of the plane with the wheel, the shuriken, and the staple constructed above.
\end{lemma}

What remains is to show that this intended tiling is the only way our three tiles can tile the plane.

\subsection{Angle Structure}

We start with a few definitions and observations on the angles of the tiles.

Call an angle \defn{clean} (and color it \CLEAN{purple})
if it is an integer multiple of $\CLEAN{\tfrac{\pi}{2n}}$.
The sum of clean angles is clean,
and the sum of clean angles and one nonclean angle is not clean.

The vertices of the convex $4n$-gon, which we call \CLEAN{\defn{corners}}, have clean convex angle $\CLEAN{\pi(1-\tfrac{1}{2n})}$. The matching shuriken reflex \CLEAN{\defn{anticorners}} have a matching defect $\CLEAN{\pi(1-\tfrac{1}{2n})}$, and each of the four concave chains are connected at the convex \CLEAN{\defn{tip}} vertices, which have a clean convex angle of $\CLEAN{\tfrac{\pi}{n}}$.

Define angles $\alpha = \tfrac{\pi}2 - 2 \epsilon$ and
$\beta = \tfrac{\pi}2 - \epsilon$, and pick $\epsilon = \tfrac{\pi}{16}$.%
\footnote{Other choices of $\epsilon$ also work; the choice here is so that
all edge lengths can be expressed by radical expressions.}
These angles and their combinations are not clean:

\begin{property} \label{lem:clean}
  For any $\theta_1,\theta_2 \in \{\alpha,\beta\}$,
  neither $\theta_1$ nor $\theta_1 + \theta_2$
  is clean.
\end{property}
\begin{proof}
  The relevant angles are
  $\alpha = \tfrac{\pi}2 - 2 \epsilon = \tfrac{3\pi}8$,
  $\beta = \tfrac{\pi}2 - \epsilon = \tfrac{7\pi}{16}$,
  $2\alpha = \tfrac{3\pi}4$,
  $2\beta = \tfrac{7\pi}8$, and
  $\alpha + \beta = \tfrac{13\pi}{16}$,
  which all have a doubly even denominator (divisible by~$4$).
  Because $n$ is odd, these numbers cannot be equal to $\CLEAN{\tfrac{i\pi}{2n}}$ for any integer $i$, so none of these angles are clean.
\end{proof}

\begin{table}[h!]
  \centering
  \begin{tabular}{lll}
    Shape & angle of convex vertices & defect of reflex vertices \\
    \hline
    staple & $\beta$ & $2 \alpha$ \\
    shuriken & $\alpha, \CLEAN{\tfrac{\pi}n}$ & $\beta, 2\beta, \CLEAN{\pi(1 - \tfrac{1}{2n})}$ \\
    wheel & $\alpha, \beta, \CLEAN{\pi(1 - \tfrac{1}{2n})}$ & $\alpha, \beta$ \\
  \end{tabular}
  \caption{Angles used in the wheel, shuriken, and staple.
    For convex vertices, we give the interior angle, while for reflex vertices,
    we give the defect ($360^\circ$ minor the interior angle).}
  \label{tab:angles}
\end{table}

Table~\ref{tab:angles} lists the angles used by each polygon, colored to indicate which are clean.
Angles $\alpha,\beta$ are all a bit less than $90^\circ$,
so a sum of two of them is a bit less than $180^\circ$,
and a sum of three of them is a bit less than $270^\circ$.
More precisely:

\begin{lemma} \label{lem:sums}
  For any $\theta_1,\theta_2,\theta_3 \in \{\alpha,\beta\}$,
  $\theta_1 \in (\tfrac{3}{8}\pi,\tfrac{7}{16}\pi)$,
  $\theta_1 + \theta_2 \in (\tfrac{3}{4}\pi,\tfrac{7}{8}\pi)$ which is $<\pi$, and
  $\theta_1 + \theta_2 + \theta_3 \in (\tfrac{9}{8}\pi,\tfrac{21}{16} \pi)$ which is $>\pi$.
  Note that these intervals are disjoint,
  so the value of a sum $\sum_i \theta_i$ with at most three terms
  determines the number of terms in the sum.
\end{lemma}


\subsection{Edge Lengths}
\label{sec:edge lengths}

We design the edge lengths of the tweedledee and tweedledum
in Figures~\ref{fig:tweedledee} and~\ref{fig:tweedledum}
so that the near-vertical and near-horizontal edges
all have the same length, which we call~$1$,
and the total horizontal traversal is exactly~$4$.
Equivalently, we choose the sequence of edge lengths for either tweedle to be
\begin{align*}
  & \big\langle 2-\sin \beta +\cos \alpha, ~
  1, ~ 1, ~ 2(\sin \alpha+\cos \beta), ~
  1, ~ 1, ~
  2-\sin \beta +\cos \alpha \big\rangle
  \\
  ={} & \big\langle 2-\cos \epsilon +\sin 2\epsilon, ~
  1, ~ 1, ~ 2(\cos 2\epsilon+\sin \epsilon), ~
  1, ~ 1, ~
  2-\cos \epsilon +\sin 2\epsilon \big\rangle.
\end{align*}
Given our choice of $\epsilon = \tfrac{\pi}{16}$,
\begin{align*}
  \cos \epsilon &= \frac{1}{2}\sqrt{2+\sqrt{2+\sqrt{2}}}, &
  \sin \epsilon &= \frac{1}{2}\sqrt{2-\sqrt{2+\sqrt{2}}}, &
  \cos 2\epsilon &= \frac{1}{2}\sqrt{2+\sqrt{2}}, &
  \sin 2\epsilon &= \frac{1}{2}\sqrt{2-\sqrt{2}}
\end{align*}
are all radical expressions,
so all our edge lengths can be expressed by radical expressions as well.

\subsection{Forced Tiling Structure}

\begin{lemma} \label{lem:staple}
  The staple does not tile the plane (allowing or forbidding reflections).
\end{lemma}

\begin{proof}
  The staple has convex vertices just of angle $\beta$
  and one reflex vertex of defect $2\alpha$.
  In any tiling of staples, every reflex vertex must have its defect $2\alpha$
  filled by some convex angles.
  But $2\alpha = \tfrac{3}{4}\pi$,
  so by Lemma~\ref{lem:sums}, it must be filled by exactly two convex angles.
  But $2\beta = \frac{7}{8}\pi > 2\alpha$,
  so two of the convex angles do not fit.
  %
  Thus it is impossible to exactly fill the deficit.
\end{proof}

\begin{lemma} \label{lem:staple+shuriken}
  The staple and shuriken do not tile the plane
  (allowing or forbidding reflections).
\end{lemma}

\begin{proof}
  By Lemma~\ref{lem:staple}, any tiling with staples and shurikens has a shuriken.
  Any shuriken has a reflex \CLEAN{anticorner} of defect
  $\CLEAN{\pi(1 - \tfrac{1}{2n})}$,
  which is clean and less than~$\pi$.
  This defect must be filled by some convex angles, of which we have three:
  $\alpha, \beta, \CLEAN{\tfrac{\pi}n}$.
  By Lemma~\ref{lem:sums}, this defect can be filled by at most two angles
  $\in \{\alpha,\beta\}$.
  But by Property~\ref{lem:clean},
  summing one or two of these angles is not clean,
  while the remaining convex angle $\CLEAN{\tfrac{\pi}n}$
  and the target sum $\CLEAN{\pi(1 - \tfrac{1}{2n})}$ are.
  Thus we cannot use any angles $\in \{\alpha,\beta\}$ to fill the
  deficit, leaving only the convex angle $\CLEAN{\tfrac{\pi}n}$.
  But $\CLEAN{\tfrac{\pi}n}$ is an even multiple of $\CLEAN{\tfrac{\pi}{2n}}$,
  while the target sum is an odd multiple of $\CLEAN{\tfrac{\pi}{2n}}$.
  Thus it is impossible to exactly fill the deficit.
\end{proof}

\begin{lemma} \label{lem:staple+shuriken+wheel}
  Any tiling of the plane with staples, shurikens, and wheels
  (allowing or forbidding reflections)
  must consist of an infinite square grid of wheels
  that corresponds to a Wang tiling.
\end{lemma}

\begin{proof}
  By Lemma~\ref{lem:staple+shuriken}, any tiling with staples, shurikens, and
  wheels has a wheel.
  Any wheel has a convex \CLEAN{corner} of angle $\CLEAN{\pi(1 - \tfrac{1}{2n})}$,
  which has a deficit of $\CLEAN{\pi(1 + \tfrac{1}{2n})}$,
  which is clean.
  We claim that this deficit can be filled in exactly two ways:
  one reflex \CLEAN{anticorner} of a shuriken of deficit $\CLEAN{\pi(1 - \tfrac{1}{2n})}$,
  or one convex \CLEAN{tip} vertex of a shuriken of angle $\CLEAN{\tfrac{\pi}n}$
  and one convex \CLEAN{corner} of a wheel of angle $\CLEAN{\pi(1 - \tfrac{1}{2n})}$.

  Because the target deficit is $> \pi$, we need to consider both convex and
  reflex angles as well as flat edges (of angle $\CLEAN{\pi}$) for possible fillings.
  First consider the unclean angles starting with reflex angles of deficit $\alpha, \beta, 2\alpha, 2\beta$. But $\alpha<\beta<2\alpha<2\beta = \tfrac{7}{8}\pi = \pi(1-\tfrac{1}{8}) < \CLEAN{\pi(1-\tfrac{1}{2n})}$ for any $n\geq 5$.
  Thus none of the unclean reflex angles can be used to fill the deficit.
  The only remaining reflex angle that can fill the deficit is a shuriken \CLEAN{anticorner} of deficit $\CLEAN{\pi(1 - \tfrac{1}{2n})}$, and
%
%
  if we use that angle, we are done.

  Next consider the unclean convex angles $\alpha$ and $\beta$. 
  Because $\CLEAN{\pi(1 + \tfrac{1}{2n})} < \tfrac{9}{8}\pi$ for any $n \geq 5$, by Lemma~\ref{lem:sums} this deficit can be filled by at most two angles $\in \{\alpha,\beta\}$.
  But by Lemma~\ref{lem:clean},
  summing one or two of these angles is not clean,
  while the remaining convex angles
  $\CLEAN{\tfrac{\pi}n}, \CLEAN{\pi(1-\tfrac{1}{2n})}$,
  flat edge of angle $\CLEAN{\pi}$,
  and the target sum $\CLEAN{\pi(1 + \tfrac{1}{2n})}$ are all clean.
  Thus we cannot use any angles $\in \{\alpha,\beta\}$ to fill the deficit.

  This leaves only the flat edge of angle $\CLEAN{\pi}$ and two convex angles: the shuriken \CLEAN{tip} of angle $\CLEAN{\tfrac{\pi}n}$ and the wheel \CLEAN{corner} of angle $\CLEAN{\pi(1-\tfrac{1}{2n})}$.
  If we used only copies of the \CLEAN{tip} $\CLEAN{\tfrac{\pi}n}$,
  we would get an even multiple of $\CLEAN{\tfrac{\pi}{2n}}$,
  but the target sum $\CLEAN{\pi(1 + \tfrac{1}{2n})}$
  is an odd multiple of $\CLEAN{\tfrac{\pi}{2n}}$.
  Using a flat edge $\CLEAN\pi$ would leave a gap of angle $\CLEAN{\tfrac{\pi}{2n}}$, which is too small to fill with any of the available angles.
  Finally, using the wheel corner $\CLEAN{\pi(1-\tfrac{1}{2n})}$ (gluing a second wheel to the first) will leave a gap of angle $\CLEAN{\tfrac{\pi}{n}}$, which can only be filled by the \CLEAN{tip} angle $\CLEAN{\tfrac{\pi}{n}}$.


  Thus we have shown that, in all cases, the wheel's convex angle $\CLEAN{\pi(1 - \tfrac{1}{2n})}$ has to be matched with an \CLEAN{anticorner} or a \CLEAN{tip} of the shuriken.
  Furthermore, an edge of the wheel adjacent to a \CLEAN{corner} must be glued to an edge of the concave chain adjacent to a \CLEAN{tip} or \CLEAN{anticorner} of the shuriken.
  We can follow the path of both the wheel and the concave chain of the shuriken and observe that the $n-2$ \CLEAN{anticorners} and the two \CLEAN{tips} of that concave chain will be glued to consecutive \CLEAN{corners} of the wheel. 

At the end of the concave chain of the shuriken, we find a \CLEAN{tip} glued to the \CLEAN{corner} of a wheel, leaving a deficit of $\CLEAN{\pi(1 - \tfrac{1}{2n})}$ to fill. As shown in Lemma~\ref{lem:staple+shuriken}, this can only be filled by another wheel \CLEAN{corner}. Therefore, the surround of a wheel must be filled by an alternating sequence of wheels and shurikens (omitting the small gaps left between the tweedles and the notches, which are filled by staples). 

Pick one wheel $T$ in the tiling, translate the tiling so that its center\footnote{Define the \defn{center} of a wheel to be the center of gravity of its $4n$-gon, without adornments.} is at the origin $(0,0)$, and rotate the plane so that edge $i$ of its $4n$-gon is glued to another wheel at a horizontal edge of the $4n$-gon, with $T$ below that edge.
Also rescale so that the width of a wheel's $4n$-gon (the distance between parallel edges, without adornments) is $1$.
Then the wheel adjacent to edge $i$ of $T$ will have its center at coordinate $(0,1)$. 
Following the boundary of $T$ clockwise, we find a shuriken glued to the edges $i+1,\ldots,i+n-1$ of $T$, and a wheel glued to the edge $i+n$ of $T$. 
The wheel glued to edge $i+n$ has its center at coordinates $(1,0)$. Continuing this reasoning, we find that the tiling is a grid of wheels with centers on all integer lattice points. The shurikens and staples ensure that all spaces are filled, and the tweedles ensure that the tiles are compatible as in a Wang tiling. Therefore, for any tiling with the three tiles, we can produce a tiling of the plane with the original Wang tiles.
\end{proof}

\subsection{Undecidability}

The angles of the polygons, as listed in Table~\ref{tab:angles},
are all rational multiples of~$\pi$.
The edge lengths of the polygons, as listed in Section~\ref{sec:edge lengths},
can all be expressed as radical expressions.
We call polygons with such angles and edge lengths \defn{nice}.

\begin{theorem}
  Given $n$ Wang tiles, we can construct three nice polygons
  that can tile the plane (allowing or forbidding reflections)
  if and only if the Wang tiles can tile the plane.
\end{theorem}

\begin{proof}
  Combine Lemma~\ref{lem:build Wang} (if) and
  Lemma~\ref{lem:staple+shuriken+wheel} (only if).
\end{proof}

Combining this with membership of the tiling problem in co-RE (to be proved in the next section), we obtain:

\begin{corollary}[Nice form of Corollary~\ref{cor:intro:complete}]
  Given three nice polygons in the plane, deciding whether they tile the plane is co-RE-complete and thus undecidable.
\end{corollary}

In a recent paper, Greenfeld and Tao
\cite{greenfeld2024undecidabilitytranslationalmonotilings}
consider a generalized version of the tiling problem, where only a periodic subset of space needs to be covered by the tiles. In our reduction, the union of the shurikens form a periodic subset of $\mathbb R^2$, and so does its complement. Thus, tiling the complements of the shurikens with the two remaining tiles is undecidable:

\begin{corollary}[Stronger form of Corollary~\ref{cor:intro:two}]
  Given two nice polygons, deciding whether they tile a periodic subset of the plane is co-RE-complete and thus undecidable.
\end{corollary}

By plugging in Wang tiles that simulate a universal Turing machine,
such as Robinson's 36 Wang tiles \cite{Robinson1971}, we also obtain
undecidability of tiling completion with three specific tiles:

\begin{corollary}[Stronger form of Corollary~\ref{cor:intro:completion}]
  There exist three fixed tiles for which completing a given finite partial tiling is co-RE-complete and thus undecidable.
\end{corollary}

\section{Membership in Co-RE}
\label{sec:co-RE}

The previous sections show the co-RE-hardness of the tiling problem for three given tiles, and the tiling completion problem for three fixed tiles.
We counterbalance these intractability results by showing that the tiling problem and tiling completion problem are in co-RE, for any finite set of prototiles.

\subsection{Model: Computable Polygons}

To make these positive results as strong as possible, we use the weakest possible model of computation.
We use a standard Turing machine, and
represent polygons by their sequences of angles and edge lengths
(equivalently, instructions in the Logo/Turtle graphics language),
which need not be given explicitly but can be computed to any desired precision.
More precisely, we assume the input polygons are ``computable'' in the following sense:

\begin{definition}
  A real number $a$ is \defn{computable} \cite{Brattka2008} if there is a Turing machine $T_a$ that, given a natural number $n$, outputs an integer $T_a(n)$ such that $\tfrac{T_a(n)-1}n \leq a \leq \tfrac{T_a(n)+1}n$.
  A polygon is \defn{computable} if it is promised to be simple and closed, and its $n$ angles and edge lengths are computable.
\end{definition}

Computability is likely the most general representation of real numbers that is still usable for our problem. Computable numbers include all rational numbers, algebraic numbers, and transcendental numbers that can be computed to any desired precision. In particular they are closed under addition, subtraction, multiplication, division, integer roots and even trigonometric functions:

\begin{theorem}[{\cite[Theorem 4.14]{Brattka2008}}]\label{thm:computable}
  If $x$, $y$, and $z$ are computable real numbers with $z>0$, then $x+y$, $x-y$, $xy$, $x/z$, $|x|$, $\min(x,y), \max(x,y)$, $\exp(x), \sin(x), \cos(x)$, $\log(z)$, and $\sqrt{z}$ are computable as well.
\end{theorem}

Note that some basic operations can be intractable for computable numbers. For instance, determining whether a computable number is zero, or the equality between two computable numbers, is undecidable (in fact, co-RE-complete). To see this, given a Turing machine $T$, define a number $a_T$ to have its $n$th bit after the binary point be $1$ if $T$ halts after $n$ steps, and $0$ otherwise. The number $a_T$ is computable, and is zero if and only if $T$ does not halt. 
For our problem, we can show a similar undecidability:
\begin{theorem}
  Given a single computable pentagon, determining whether it tiles the plane is co-RE-hard and thus undecidable.
\end{theorem}
\begin{proof}
  Pick a generic quadrilateral, and glue a very flat isosceles triangle to one of its edges, where the apex angle is $180^\circ - r/100$ for a given constructible number $r \in [0,1]$. The other angles and edge lengths of the triangle are constructible by constructibility of trigonometric functions. 

  The resulting pentagon tiles the plane only in the degenerate case where the triangle is degenerate, i.e., a line segment, which happens exactly when the obtuse angle of the triangle is exactly $180^\circ$. Thus the tiling problem is equivalent to testing whether $r=0$ for a given constructible number $r \in [0,1]$, which is co-RE-hard as shown above.
\end{proof}

Of course, this result is quite unsatisfying, as the reason for undecidability stems from the extreme weakness of the model and generality of the representation of the polygons, and the inability to even check locally that a tiling is valid. Yet surprisingly, in this same model, we are able to show membership in co-RE.

\subsection{Co-RE Algorithm}

The high-level idea of our algorithm is to try to build partial tilings that
cover a larger and larger disk.  If we ever fail to cover a disk,
then we know that the plane cannot be tiled; and if we never fail,
the well-known Extension Theorem (Theorem~\ref{thm:extension} below)
guarantees that the plane can be tiled.
To determine whether we can tile enough to cover a disk,
we bound the number of tiles that could possibly intersect the disk,
then enumerate all possible combinatorial ways for these tiles
to fit together, and for each, check whether the tiles fit together properly.
Checking fit is limited to tiles that share vertices, however,
so we need to take care to handle the case that there are seam lines
in the tiling where no tiles on opposite sides of the line share a vertex
(as in, e.g., the classic brick tiling).
We also avoid checking for global intersection between tiles
(because doing so is tricky in our model), opting instead to check
just locally that angles add up correctly at vertices and that
edge lengths add up correctly along edges.
Our notion of ``neat carpet'' handles both of these issues by forbidding
only local self-overlap, and guaranteeing that every boundary vertex
is either outside the specified disk or has total angle $180^\circ$
so potentially forms a seam boundary.
We are then able to show that arbitrarily large neat carpets
imply the existence of a plane tiling.

Our co-RE algorithm will in particular need to repeatedly test for equality among constructible numbers, a co-RE-complete problem.
Thus we need a way to compose co-RE decisions.
We use the following standard result (mentioned, e.g., in \cite{Rice1953}):
\begin{lemma}\label{lem:co-RE}
  Finite disjunctions and recursively enumerable conjunctions of co-RE decision problems are in co-RE.
\end{lemma}
\begin{proof}
We provide a short proof for completeness.

\textbf{Finite disjunctions:} We want to decide whether $P_1 \vee P_2 \vee \cdots \vee P_n$ is true, where each $P_i$ is in co-RE. We can run the algorithm for $P_1$, then $P_2$, and so on. If they all halt and return false, then we halt and return false. If any of them halts and returns true, then we halt and return true; and if any one never stops (equivalent to outputting true), then the disjunction is true and our algorithm never stops. This algorithm proves membership in co-RE.

\textbf{Recursively enumerable conjunctions:} We want to decide whether $P_1 \wedge P_2 \wedge \cdots$ is true where each $P_i$ is in co-RE, and the infinite list of $P_i$ is enumerated by an algorithm. In round $j$ of our algorithm, for $j=1,2,\ldots$, we run the first $j$ steps of $P_1,P_2,\ldots,P_j$ in sequence. If any of them halts and returns false, then we halt and return false. If all of them return true or never stop, then the conjunction is true and our algorithm returns true or never stops. This algorithm proves membership in co-RE.
\end{proof}






Let $\mathcal T$ be a set of prototiles,
where each tile is a (simple closed) computable polygon.
Define a \defn{carpet} to be a topological disk
produced by gluing together a finite collection of tiles from $\mathcal T$,
where every interior vertex has $360^\circ$ total angle from incident tiles.
We assume that the carpet is laid out in the Euclidean plane,
that is, every point in a tile has real coordinates,
but we allow the surface to be self-overlapping, that is,
a point of the plane might be covered by more than one tile.
A \defn{patch} is a carpet whose embedding in the plane is not self-overlapping.

A carpet can be described by its \defn{combinatorial gluing}, which specifies
(1)~the set of tiles, each of which is an instance of a prototile;
(2)~a partition of the tile vertices into coincident (glued together) points;
and
(3)~for each tile edge, the sequence of other tile edges and/or boundary
that the edge has positive-length overlap with, in order along the edge.
Call a carpet or a patch \defn{seamless} if the position of all tiles in the carpet is fully determined by the combinatorial gluing and the position of its first tile.
This notion forbids carpets whose tiles can be separated by a line along which the two sides could slide (causing an uncountable infinity of solutions).
To verify seamlessness, build the \defn{incidence graph} on the tiles of the carpet, where two tiles are connected by an edge if they share a vertex, and check that the graph is connected.
For the case of tiling completion, we can ensure that the given patch is seamless by adding a vertex along each seam, common to both adjacent tiles; as this modification to the preplaced tiles does not change their shape, it does not change the outcome of the decision.

First we show how to verify that a carpet is valid:
\begin{lemma}\label{lem:valid}
Given a combinatorial gluing of a possible seamless carpet, deciding whether it corresponds to a seamless carpet is co-RE-complete.
\end{lemma}
\begin{proof}
First we show that the problem is in co-RE. 

Start by verifying that the carpet is a topological disk,
by checking that its boundary has a single connected component that is a cycle.
Then, for every interior vertex of the carpet (not incident to a boundary edge), check that the sum of angles of the tiles incident to it is $2\pi$. Then check that the glued edges match in length. While this is trivial when two adjacent tiles match edge to edge, we need special care when a vertex is glued to the interior of an edge.

To account for this, for every pair of adjacent tiles that share a vertex, follow the straight line path along their common edge, noting the successive sequence of vertices and edges, until either (1)~the path ends at a vertex on both sides, or (2)~we reach the boundary of the patch. In the first case, check that the sum of lengths on both sides of the path are equal. In both cases, check that the order in which the vertices are encountered is compatible with the edge lengths, which reduces to testing equalities and inequalities over sums of lengths. Because the carpet is seamless, the length of all edges will be checked in this step.

The entire verification can be expressed as a conjunction of a finite number of equality ($=$) and inequality ($\leq$) tests, all of which are co-RE decisions. By Lemma~\ref{lem:co-RE}, this check is in co-RE.


Finally, to show that the problem is co-RE-hard, create a patch with three copies of a triangular prototile, rotating the prototile around one of its vertices. The patch is valid if and only if the angle of the triangle at that vertex is $120^\circ$, which is co-RE-complete by reduction from the equality of computable numbers.
\end{proof}


Define the \defn{distance} between two points in a carpet to be the Euclidean distance between those two points when the carpet is laid out in the Euclidean plane (note that this embedding may be self-overlapping, and this distance is no larger than the intrinsic distance within the carpet).
Call a vertex of a carpet \defn{neat} if it is either interior to the carpet and surrounded by tiles summing up to an angle of $2\pi$, or is on the boundary of the carpet and is surrounded by contiguous tiles summing up to an angle of $\pi$.
A carpet is \defn{neat within radius $<r$} if every vertex at distance $<r$ from the origin is neat.
In an \defn{anchored} carpet,
we assume one \defn{anchored} tile in the combinatorial gluing
has been chosen to be placed with its center of gravity at the origin,
and with a canonical rotation (say, matching its prototile);

\begin{lemma} \label{lem:neat-co-RE}
Given an anchored carpet and a computable positive number $r$, deciding whether it is neat within radius $<r$ is in co-RE.
\end{lemma}
\begin{proof}
The coordinates of each vertex of the carpet laid out in the plane
(relative to the anchored tile)
are computable by Theorem~\ref{thm:computable}, using trigonometric functions, and so is the vertex's distance from the origin.
Checking that each vertex is neat is an equality test of a sum of angles, which is co-RE.
For each vertex, we verify that it is either neat or that its distance to the origin is $\geq r$.
This is a disjunction of two co-RE problems, so by Lemma~\ref{lem:co-RE}, is in co-RE.
Checking this property for all vertices is a finite conjunction of co-RE problems, so by Lemma~\ref{lem:co-RE}, is in co-RE.
\end{proof}

Let $\rho$ be the maximum ``radius'' of tiles in $\mathcal T$, meaning that a disk of radius $\rho$ centered at the center of gravity of each tile in $\mathcal T$ covers that tile.
Let $A_{\min}$ be the minimum area of a tile in $\mathcal T$.
Let $D_r$ denote the disk of radius $r$ centered at the origin.
\begin{lemma} \label{lem:tiles-bound}
  If $\mathcal T$ can tile the plane, then for any $r > 0$, (a) there is an anchored patch with a finite number $N(r)$ of tiles that covers the disk $D_r$, and (b) there is an anchored seamless patch with $\leq N(r)$ tiles that is neat within radius $<r$.

  If a seamless patch $P$ can be extended to tile the plane, then for any $r > 0$, (a) there is an anchored patch containing $P$ with a finite number $|P|+N(r)$ of tiles that covers the disk $D_r$, and (b) there is a seamless patch containing $P$ with $\leq |P|+ N(r)$ tiles that is neat within radius $<r$.
\end{lemma}
\begin{proof}
  For the (a) statements, translate (and rotate) the tiling so that one of its tiles is anchored; and if we are given a seamless patch $P$, choose to anchor one of its tiles.
  Consider the disk $D_{r+2\rho}$ of radius $r+2\rho$ centered at the origin (the anchored tile's center of gravity).
  Take all the tiles in the plane tiling that are fully inside $D_{r+2\rho}$, and take all the tiles of $P$ (if given).
  Because the tiles do not overlap and are each of area $\geq A_{\min}$, there are at most $\pi(r+2\rho)^2/A_{\min}$ tiles inside $D_{r+2\rho}$.
  Take the connected component that contains the origin (and thus the tiles of $P$, if given), which only decreases the number of tiles.
  This is an anchored patch that covers the smaller disk $D_r$.

  For the (b) statements, build the incidence graph on the tiles of the carpet, where two tiles are connected by an edge if they share a vertex.
  If this graph is connected, then the patch is seamless.
  If this graph is disconnected, it is because of seams.
  Seam lines cannot intersect: otherwise, their intersection point is a vertex common to both sides of the seams.
  Thus, cutting the patch along all seams, or equivalently taking one connected component of the incidence graph, will create neat vertices on the new boundary, where the seams were.
  Take the component that contains a tile covering the origin. This is a seamless patch that is neat within radius $<r$.
\end{proof}

\begin{lemma} \label{lem:large-neat}
  Given prototiles $\mathcal T$ and a computable radius $r$, deciding whether there is an anchored carpet that is neat within radius $<r$ is in co-RE.
  If there is no such carpet, then $\mathcal T$ cannot tile the plane.

  Given prototiles $\mathcal T$, a seamless patch $P$, and a computable radius $r$, deciding whether $P$ can be extended to a carpet that is neat within radius $<r$ is in co-RE.
  If there is no such carpet, then $P$ cannot be extended to tile the plane.
\end{lemma}
\begin{proof}
  Enumerate all possible combinatorial gluings of carpets, or carpets extending $P$, with $\leq N(r)$ (extra) tiles.
  Keep only the carpets that are seamless.
  Try all possible anchored tiles, restricting to a tile in $P$ if given.

  Check that the combinatorial gluing corresponds to an actual carpet, which is in co-RE by Lemma~\ref{lem:valid}.
  Check that the carpet is neat within radius $<r$, which is in co-RE by Lemma~\ref{lem:neat-co-RE}.

  If any of these carpets satisfy all checks, then there is a neat carpet within radius $<r$.
  If none are valid, then by Lemma~\ref{lem:tiles-bound}, $\mathcal T$ cannot tile the plane, or $P$ cannot be extended to tile the plane.
  This is a finite disjunction of co-RE problems, so it is in co-RE.
\end{proof}

\begin{lemma}\label{carpettrim}
  If an anchored carpet is neat within radius $< r+2\rho$, then the carpet contains an anchored patch that is neat within radius $<r$.
\end{lemma}
\begin{proof}
  Consider the intersection between the carpet $C$ and the disk $D_{r+2\rho}$, which might have multiple connected components and/or be self-overlapping.
  Pick the component $S$ that contains the anchored tile. The boundary of $S$ is composed of straight lines (the edges of the tiles connected by neat vertices) and circular arcs (portions of the boundary of $D_{r+2\rho}$), connected together at convex angles. Thus $S$ is convex, so non-self-overlapping.

  Back in the carpet, remove all tiles that are not entirely contained in $D_{r+2\rho}$. Again, this possibly results in several connected components.
  Retain the component $C'$ that contains the anchored tile,
  where connectivity is defined by interior paths,
  so that $C'$ does not have pinch points.
  By construction, this component is contained in $S$ and thus is not self-overlapping.
  Also $C'$ is a topological disk: a hole in $C'$ could only come from a
  removed tile, which must touch the outside of $D_{r+2\rho}$,
  contradicting that it is surrounded by~$C'$ (by planarity).
  Therefore, $C'$ is a patch.

  Finally, because all tiles removed have diameter $\leq 2\rho$, none of the deleted tiles intersect the disk $D_r$, and therefore all vertices within radius $<r$ remain untouched. Therefore the patch $C'$ is neat within radius $<r$. 
\end{proof}

\begin{lemma} \label{lem:neat2patch}
  If there exists an anchored carpet $C$ that is neat within radius $<r+2\rho$, then there exists a patch that covers the disk $D_{r/2}$. Furthermore, that patch contains all tiles of $C$ that intersect $D_{r/2}$.
\end{lemma}

\begin{proof}
  Suppose we have a carpet that is neat within radius $<r+2\rho$. By Lemma~\ref{carpettrim}, we have a patch that is neat within radius $<r$.
  The intersection of the patch and the disk $D_r$ is a disk cut by noncrossing chords;
  refer to Figure~\ref{fig:neat carpet disks}.
  Note that any chord of $D_r$ that intersects $D_{r/2}$ cuts off an arc of angle $\tfrac{2}{3} \pi$ from $D_r$. Because the chords defined by the patch do not intersect in $D_r$, at most two of them intersect~$D_{r/2}$.

  \begin{figure}
    \centering
    \includegraphics[width=\linewidth]{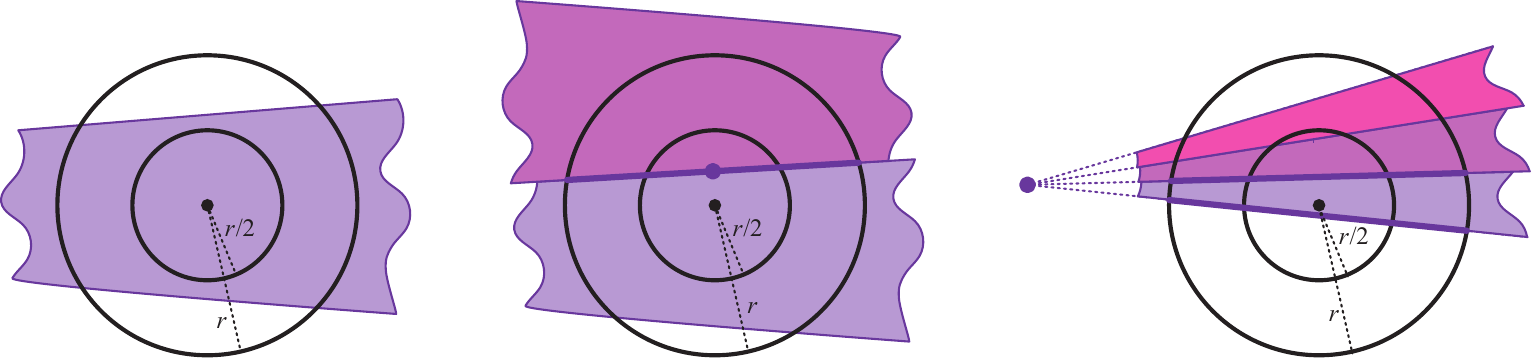}
    \caption{A neat patch within $<r$, and how it can interact with
      the smaller disk $D_{r/2}$.  From left to right:
      no chords, one chord, and two chords.}
    \label{fig:neat carpet disks}
  \end{figure}

  If no chord intersects $D_{r/2}$, then the patch covers $D_{r/2}$, and we are done. 

  If exactly one chord intersects $D_{r/2}$, then the patch is contained in a half-plane that contains the origin and is bounded by the extension of the cord. Rotate a copy of the patch by $180^\circ$ about the center of the cord, and glue the two pieces together. This produces a patch that covers $D_{r/2}$. 

  If two chords intersect $D_{r/2}$, then the patch is contained in a wedge defined by the extensions of the two chords.
  Let $q$ be the apex of the wedge, and assume $q$ is on the $x$ axis, with one of the chords above the $x$ axis and the other chord below.
  Assume without loss of generality that the chord below is the longest one.
  Repeatedly stitch copies of the patch by rotating it about~$q$, upward, gluing the longer chord of the previous copy to the shorter chord of the next copy, until the upper half of $D_{r/2}$ is covered.
  Now repeat the procedure for a single chord.
  The resulting patch covers the disk of radius $r/2$.
\end{proof}

\begin{theorem}[{Extension Theorem \cite[p.~151]{Gruenbaum-Shephard}}]
  \label{thm:extension}
Given a finite collection $\mathcal T$ of prototiles, if they tile arbitrarily large disks, then they admit a tiling of the plane.

Given a finite collection $\mathcal T$ of prototiles, and a patch $P$ using tiles of $\mathcal T$, if $P$ can be extended to cover arbitrarily large disks centered in $P$, then $P$ can be completed to tile the plane.%
\footnote{Although \cite[p.~151]{Gruenbaum-Shephard} does not explicitly state the completion version of the theorem, the same proof establishes both.}
\end{theorem}  

We translate the above theorem to seamless anchored patches.
\begin{lemma} \label{lem:neatextension}
  Given a collection $\mathcal T$ of prototiles, if there exist anchored carpets that are neat within radius $<r$ for arbitrarily large $r$, then $\mathcal T$ admits a tiling of the plane.

  Given a collection $\mathcal T$ of prototiles, and given an anchored patch $P$ using tiles of $\mathcal T$, if $P$ can be extended to a carpet that is neat within radius $<r$ for arbitrarily large $r$, then $P$ can be completed to tile the plane.
\end{lemma}
\begin{proof}
  If there exists an anchored carpet or a carpet extending $P$ that is neat within radius $<r$ for arbitrarily large $r$, then by Lemma~\ref{lem:neat2patch}, there exists a patch (extending $P$) that covers the disk of radius $r/4$ centered at the origin for arbitrarily large $r$. By the Extension Theorem, the prototiles admit a tiling of the plane.
\end{proof}

\begin{theorem}[Precise form of Theorem~\ref{thm:intro:coRE}] \label{thm:tiling-co-RE}
  Given a set $\mathcal T$ of $k$ polygons in our model,
  deciding whether they tile the plane is in co-RE.
  Also given a patch $P$ of tiles from $\mathcal T$, deciding whether $P$ can be completed to tile the plane is in co-RE.
\end{theorem}

\begin{proof}
  For every positive integer $k$, set $r=k\rho$ and use Lemma~\ref{lem:large-neat} to determine whether there exists an anchored carpet that is neat within radius $<r$, or whether $P$ can be completed to produce such a carpet.
  This is a recursively enumerable disjunction of co-RE problems, so by Lemma~\ref{lem:co-RE} is in co-RE. 
  By Lemma~\ref{lem:neatextension}, if all of these problems output true, then the polygons tile the plane or complete $P$ to tile the plane. By Lemma~\ref{lem:large-neat}, if any of these problems outputs false, then the polygons do not tile the plane, or cannot complete $P$ to tile the plane.
\end{proof}

\section{Finding Periodic Tilings}
\label{sec:periodic}

We now consider the problem of determining whether prototiles $\mathcal T$ can tile the plane periodically. As seen in the previous section, it will be necessary to restrict our representation model slightly, to guarantee that even checking a patch is decidable. For our result, we restrict all edge lengths and angles (in degrees) to be algebraic numbers. 
Because comparison between sums of algebraic numbers is decidable \cite{SepBound}, we can revisit Lemma~\ref{lem:valid} and note that the algorithm uses only such comparisons:

\begin{lemma}\label{lem:valid-comparable}
  Given a combinatorial gluing of a possible seamless carpet, where edge lengths and angles of prototiles are algebraic numbers,
  determining whether it is valid is decidable.
\end{lemma}

A \defn{periodic tiling} $T$ using prototiles $\mathcal T$ is a tiling that has two linearly independent translational symmetries (say, $\vec{a}$ and $\vec{b}$) that act on the tiles of $T$. 
What this means is that applying the translation vector $\vec{a}$ (respectively $\vec{b}$) on any tile $t\in T$ produces another tile of $T$.
The symmetry group generated by $\vec{a}$ and $\vec{b}$ decomposes the set of tiles of $T$ into equivalence classes, also called \emph{orbits}, where two tiles are in the same class if there is a symmetry in the group (an integer linear combination $i \vec{a} + j \vec{b}$ for some $i,j \in \mathbb Z$) that matches one to the other. 

The tiling can then be described by a \defn{fundamental domain} \cite[p.~55]{Gruenbaum-Shephard} for the action of these symmetries. The fundamental domain is a patch $S$ containing exactly one tile from each orbit which, when fused together in one supertile, tiles the plane isohedrally using the same translations. Its boundary can be decomposed into six pieces, denoted by $A,B,C,\bar{A},\bar{B},\bar{C}$, counterclockwise, where $\bar{A}$, $\bar{B}$, and $\bar{C}$ are translations of $A$ by the action of $\vec{a}$, $B$ by the action of $\vec{b}$, and $C$ by the action of $\vec{b}-\vec{a}$, respectively \cite{G-BN1991}.

The \defn{existential theory of the reals} is the decision problem
$\exists x_1 \in \mathbb R : \cdots : \exists x_n \in \mathbb R : F(x_1,\ldots,x_n)$,
where $F$ is a polynomial-size quantifier-free formula involving (in)equalities of polynomials with integer coordinates  \cite{Basu2006}.
Complexity class $\exists \mathbb R$ \cite{ExistsR-survey}
is the class of problems that can be reduced to this decision problem,
or equivalently,
that can be solved by a polynomial-time algorithm that can guess real values nondeterministically \cite{Erickson-ER}.

\begin{lemma}
  Given prototiles $\mathcal T$ where all tile edge lengths and angles are algebraic numbers, and given a number $k$, determining whether $\mathcal T$ tiles the plane periodically with a fundamental domain of size $\leq k$ is in $\exists\mathbb R$ and thus in PSPACE.
\end{lemma}
\begin{proof}
As in the previous section, enumerate every possible combinatorial gluing of a carpet. Check whether it is seamless. If not, guess the real shift between the two parts of each seam, and add a common vertex to two tile edges that are adjacent across the seam, thus stitching all into a seamless carpet. 
For each carpet, guess the six points that divide its outside face into six pieces. Check that the sequence of edge lengths and angles match appropriately.
\end{proof}

By running this algorithm for every $k\in\mathbb N$, we obtain that periodic tiling is in RE:

\begin{corollary}
  Given prototiles $\mathcal T$ where all tile edge lengths and angles are algebraic numbers, determining whether they admit a periodic tiling of the plane is in RE.
\end{corollary}

Combining this membership in RE with the membership in co-RE from Theorem~\ref{thm:tiling-co-RE}, we obtain a decidable case:

\begin{corollary}
  Given prototiles $\mathcal T$ where all tile edge lengths and angles are algebraic numbers,
  and given the promise that any tiling with $\mathcal T$ is periodic,
  determining whether $\mathcal T$ tiles the plane is decidable.
\end{corollary}

For example, tiling by translation with a single prototile is decidable
because such tiling must be periodic \cite{G-BN1991}.

%

\section*{Acknowledgements}
This research was initiated at the Second Adriatic Workshop on Graphs and Probability, and continued at the Eleventh Annual Workshop on Geometry and Graphs. The authors thank all participants of both workshops for providing a stimulating research environment.

\bibliography{tiling}
\bibliographystyle{alpha}

\end{document}